\documentclass[onecolumn]{IEEEtran}
\usepackage{hyperref}
\usepackage{amsmath,amsthm,amssymb}
\usepackage{color}
\usepackage{url}
\usepackage{balance}
\usepackage{graphicx} 
\usepackage{mathrsfs}
\usepackage{epsfig}
\usepackage{verbatim}
\usepackage{setspace}

\newtheorem{theorem}{Theorem}
\newtheorem{lemma}[theorem]{Lemma}

\newtheorem{corollary}[theorem]{Corollary}
\newtheorem{definition}[theorem]{Definition}
\newtheorem{claim}[theorem]{Claim}

\renewcommand{\vec}[1]{\mathbf{#1}}
\title{Writing on a Dirty Paper in the presence of Jamming}
\author{Amitalok~J.~Budkuley, Bikash~Kumar~Dey and Vinod~M.~Prabhakaran\\
   Emails: \{amitalok, bikash\}@ee.iitb.ac.in, vinodmp@tifr.res.in}
\begin{document}
\maketitle 
\begin{abstract}
In this paper, the problem of writing on a dirty paper in the presence of jamming is examined. We consider an AWGN channel with an additive white Gaussian state and an additive adversarial jammer. The state is assumed to be known non-causally to the encoder and the jammer but not to the decoder. The capacity of the channel in the presence of a jammer is determined. A surprising result that this capacity is equal to the capacity of a relaxed version of the problem, where the state is also known non-causally to the decoder, is proved.
\end{abstract}
\section{Introduction}
We study the problem of `writing on a dirty paper'~\cite{costa} in the presence of an intelligent jammer. In the absence of a jammer, the problem of writing on a dirty paper is a channel coding problem where the transmitter wants to send a message over an AWGN channel with additive interference, and the interference is known to the transmitter non-causally. Under the assumption of only the receiver knowing the interference, we know that the receiver can completely cancel it out. However, in a celebrated result~\cite{costa} Costa showed that even in the scenario when only the transmitter (and not the receiver) possessed the knowledge of the interference, it can be completely mitigated using a novel coding scheme called \textit{Dirty Paper Coding (DPC)}. Through this coding scheme, it was shown that the capacity of this channel is equal to the capacity of one with no state. 

In our problem, we have an additional entity, viz., an adversarial jammer in the communication setup. Not only is the user aware of the additive interference in the channel as in~\cite{costa}, but so also is the jammer. Through the choice of its channel input the user aims to maximize its communication rate over the AWGN channel. In contrast, the adversarial jammer aims to minimize the user communication rate through an appropriately chosen jamming signal. Thus, we have a non-cooperative game between the user and the additive jammer. Since both the user and the jammer compete over the same quantity, this non-cooperative game is a zero sum game~\cite{owen}.

Jamming in communication systems has been widely studied through the use of non-cooperative game theory. Ba\c{s}ar, in~\cite{basar}, studied the problem of communicating a Gaussian random variable over an AWGN channel in the presence of a jammer. In~\cite{medard}, M\'edard studied the channel coding problem of communicating a message over an AWGN channel in the presence of a jammer, where the jammer is correlated to the user transmission. Here, the user and the jammer compete over the channel mutual information and the resulting zero sum game is shown to have an equilibrium saddle point, where, the equilibrium utility is defined as the capacity of the AWGN channel in the presence of a jammer. This game was extended to the MIMO fading scenario in~\cite{kashyap}. Shafiee et. al in~\cite{shafiee} studied this mutual information game between the user and the jammer over the AWGN Multi-Access Channels (MAC). They also studied similar formulations over the Fading AWGN MAC with jammer under different assumptions on the channel state knowledge at the users and the jammer. 

In this paper, we study the problem of writing on a dirty paper in the presence of an additive jammer, where the state is assumed to be additive white Gaussian. Also, both the user and the  jammer are assumed to have a non-causal access to this additive channel state. We formulate and analyse a zero sum mutual information game, called \textit{Costa game} in the sequel, where the user and the jammer, maximize and minimize respectively, a well motivated mutual information quantity (different from that in~\cite{medard}). It is well-known that all Nash equilibria in a zero sum game possess identical value of utility~\cite{owen}. For a channel with an adversarial jammer, define its \textit{capacity} as follows. 
\begin{definition}[Capacity]\label{def:capacity}
The \textbf{capacity} of the channel in the presence of the jammer is the unique Nash equilibrium utility of the zero sum communication game between the user and the jammer. 
\end{definition}
%
We first determine a Nash equilibrium saddle point for the Costa game, where the user strategy is the DPC strategy and the jamming strategy is i.i.d Gaussian jamming independent of state. We then determine the game utility at this equilibrium and thus, establish the capacity of the AWGN channel with additive white Gaussian state and an additive adversarial jammer under the assumption of non-causal state knowledge at the encoder and the jammer. 
We also study a relaxed version of the said problem, where, in addition to the encoder and the jammer, the receiver too has non-causal knowledge of the state. We call the corresponding zero sum game the \textit{Side Information (SI)} game. A Nash equilibrium and the corresponding game utility for this game are determined, which thence establish the capacity of the said channel. 
It is seen that the equilibrium Costa game utility equals the equilibrium SI game utility. This proves an interesting result that the capacity of the AWGN channel with additive white Gaussian state in the presence of a jammer, where the encoder and jammer have non-causal knowledge of this state, is equal to the capacity of its relaxed version, where in addition to the encoder and the jammer, the receiver too has non-causal knowledge of the state.
 
The following is the organization of the paper. In Section~\ref{sec:system:model}, we describe the system model and discuss our problem setup and the resulting non-cooperative zero sum game. We also review certain important and relevant results on the capacity of channels with state. We state the main results of this work in Section~\ref{sec:main:results}. Then, in Section~\ref{sec:costa:game:saddle:point}, we perform the analysis of the game and prove the results stated in Section~\ref{sec:main:results}. The last section is devoted towards stating certain important implications of our results as well making some overall concluding remarks.
%
\section{System Model and Problem Description}\label{sec:system:model}
\subsection{The Dirty Paper Coding Problem Setup with an Additive Jammer}
\begin{figure}[!ht]
  \begin{center}
    \includegraphics[trim=0cm 6cm 0cm 1cm, scale=0.3]{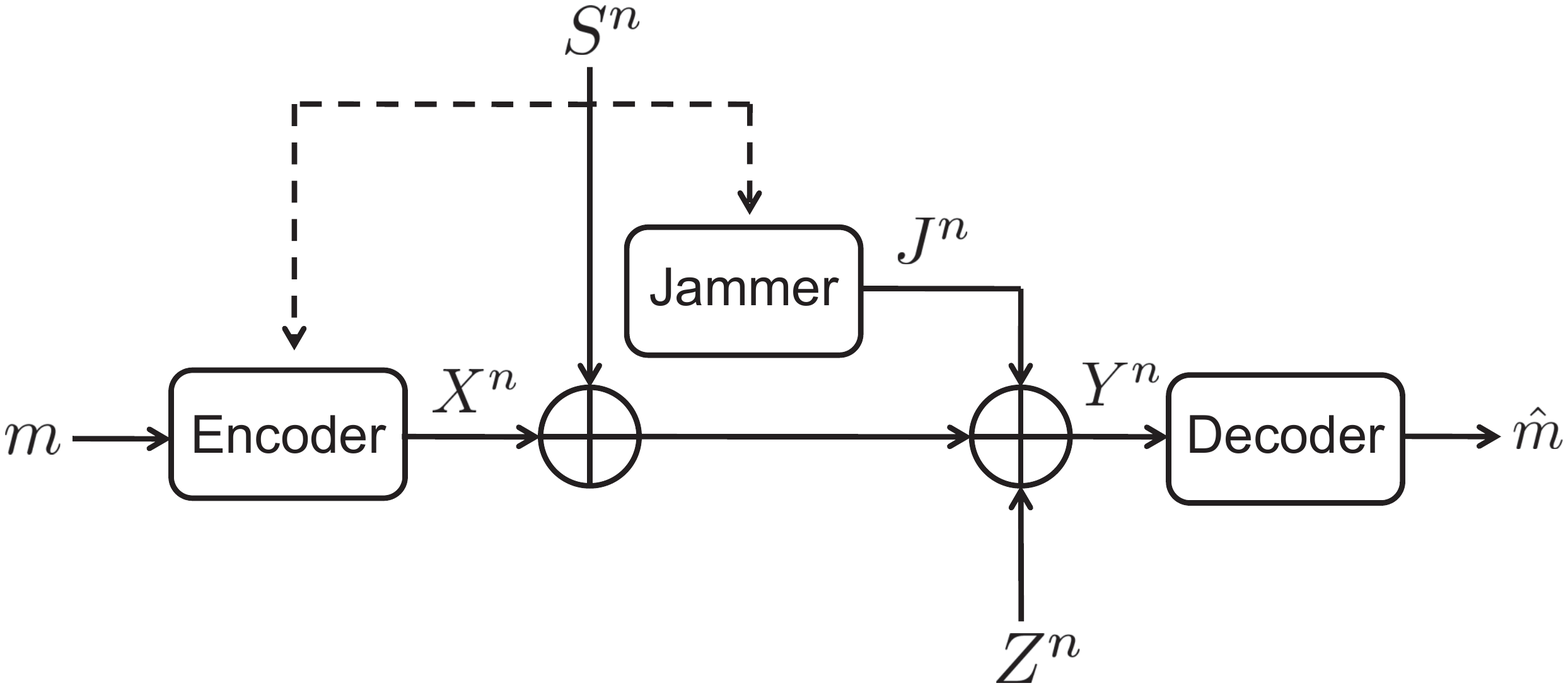}
    \caption{The Costa Game Setup}
    \label{fig:costa:game:setup}
  \end{center}
\end{figure}
Consider the communication setup in Fig.~\ref{fig:costa:game:setup}. Here, the sender wishes to send a message $m$ to the receiver through $n$ uses of a channel in the presence of a jammer. The communication channel, henceforth called the \textit{Costa} channel with a jammer, is an AWGN channel with an additive white Gaussian state and an additive adversarial jammer. Let $Y^n=(Y_1,Y_2,\ldots,Y_n)$ be the signal received at the decoder, where  
\begin{equation*}
Y^n=X^n+S^n+J^n+Z^n.
\end{equation*}
Here, $X^n$, $S^n$, $J^n$ and $Z^n$ are the user's input to the channel, white Gaussian state, jammer's input to the channel and the independent channel noise respectively. Since the state is white Gaussian, $S_i$, $i=1,2,\dots,n$ are i.i.d and each $S_i\sim\mathcal{N}(0,\sigma_S^2)$. The channel noise $Z^n$ is i.i.d Gaussian noise with each $Z_i\sim \mathcal{N}(0,\sigma^2)$. The state $S^n$ is known non-causally to only the encoder and the jammer but not the decoder. Given the message $m$ and the state $S^n$, the encoder picks a codeword $X^n(m,S^n)$ and transmits it on the channel over $n$ successive time instants. For ease of notation we drop the reference to $m$ and $S^n$. The codeword $X^n$ is so chosen as to satisfy the user power constraint $P_U$ i.e. $\sum_{i=1}^n X_i^2 \leq nP_U$. The jammer signal $J^n$, which is a function of the state $S^n$, is chosen as to satisfy the jammer power constraint $P_J$ i.e. $\sum_{i=1}^n J_i^2 \leq n\bar{P_J}$. Note that the jammer is assumed to be unable to listen to the user transmission, and hence, is conditionally independent of the user signal $X^n$ given $S^n$. 

\subsection{Capacity of the channel in the absence of a jammer}
We now review some known results on the capacity of the Costa channel in the absence of a jammer. The discrete alphabet version of the problem in Fig.~\ref{fig:costa:game:setup}, in the absence of the jammer, was solved in~\cite{gelfand}. Here, it was shown that the capacity of a channel with a random state, where the state is known non-causally at the encoder only, is given as
\begin{equation} \label{eq:GP}
C_{GP}=\max_{p(u,x|s)} I(U;Y)-I(U;S).
\end{equation}
Here, $U$ is an auxiliary random variable and the subscript $GP$ stands for Gelfand-Pinsker.

The Gaussian version of the above problem with additive white Gaussian state was studied by Costa in~\cite{costa}. Here, the surprising result, that the effect of the additive state can be completely eliminated under state knowledge only at the encoder, was proved. Thus, the the capacity of an AWGN channel with a white Gaussian state, where the state is known non-causally only at the encoder, is equal to the capacity of the standard AWGN channel with no state. This capacity, denoted by $C_C$ where the subscript $C$ stands for Costa, is 
\begin{equation*}
C_C=\frac{1}{2} \log\left(\ 1+\frac{P_U}{\sigma^2}\right).
\end{equation*}
The coding scheme to achieve this capacity is popularly known as the `Dirty Paper Coding' scheme. For a length $n$ code~\cite{elgamal-kim}, the optimal auxiliary variable $U^n$ is i.i.d with $U_i=X_i+\alpha S_i$, $\forall i=1,2,\dots,n$, where $X^n$ is i.i.d with $X_i\sim\mathcal{N}(0,P_U)$, $i=1,2,\dots,n$ and $\alpha=P_U/(P_U+\sigma^2)$. 

\subsection{The Non-Cooperative Costa Game Setup}
We model the user and jammer interaction in Fig.~\ref{fig:costa:game:setup} as a non-cooperative game. We call this game the \textit{Costa} game. In a manner similar to~\cite{medard}, the user and the jammer compete over a channel mutual information quantity. Appealing to~\eqref{eq:GP}, let this quantity be $I(U^n;Y^n)-I(U^n;S^n)$, where $U^n$ is an auxiliary random variable. The mutual information is evaluated for some joint distribution $p(s^n,u^n,x^n,j^n,y^n)=p(s^n)p(u^n,x^n|s^n)p(j^n|s^n)p(y^n|x^n,j^n)$. The user is the maximizing player and chooses $p(u^n,x^n|s^n)$ under its power constraint $P_U$ so as to maximize $I(U^n;Y^n)-I(U^n;S^n)$. The jammer is the minimizing player and chooses $p(j^n|s^n)$ under its power constraint $P_J$ so as to minimize $I(U^n;Y^n)-I(U^n;S^n)$. Note that $p(s^n)$ and $p(y^n|x^n,s^n,j^n)$ are fixed. 
%
From Definition~\ref{def:capacity}, we know that the equilibrium utility value of this game is defined as the capacity of the Costa channel in the presence of a jammer. We therefore analyse this game and determine its equilibrium utility. 
Note that the strategy spaces of the user and the jammer are both convex,
and thus the game has a Nash equilibrium~\cite{owen}.

\section{The Main Result}\label{sec:main:results}
Our main result is the characterization of an equilibrium of the Costa game and thence, the determination of the capacity of the AWGN channel with an additive state in the presence of an adversarial jammer, where the state is known non-causally to both the user and the jammer. The equilibrium involves
an i.i.d. Gaussian jamming strategy. For such a jamming strategy, the jamming
signal and the noise together act like white Gaussian noise. The resulting
channel is the standard AWGN channel with additive state analysed by Costa~\cite{costa}.
The user can employ DPC to achieve rates upto $1/2 \log(1+P_U/(P_U+P_J+\sigma^2))$ as
shown in~\cite{costa}. Such a user strategy for i.i.d. Gaussian jamming will
henceforth be referred as DPC user strategy.
We now state our main result.

\begin{theorem}\label{thm:costa:game:saddle:point}
The Dirty Paper Coding user strategy and the i.i.d. Gaussian jamming strategy
independent of the state form a Nash equilibrium of the Costa game.
\end{theorem}
%

As a consequence of the above theorem, we have the following Corollary.
\begin{corollary}\label{cor:capacity}
The capacity of the AWGN channel with an additive jammer and an additive white Gaussian state, where the state is known non-causally to both the encoder and the jammer, is 
\begin{equation}\label{eq:costa:game:saddle:point}
C_{C-J}=\frac{1}{2} \log\left(1+\frac{P_U}{\sigma^2+P_J} \right).
\end{equation}
\end{corollary}
Here, the subscript $J$ denotes the jammer. The proof of our main result will be discussed in the following section.
\section{The Costa Game Saddle Point Analysis}\label{sec:costa:game:saddle:point}
Toward the analysis of the Costa Game saddle point, we begin by studying a relaxed version of this game which we call the Side Information (SI) game. 
%
\subsection{The Side Information (SI) Game}
Most of the system setup in the SI game is identical to the setup in Fig.~\ref{fig:costa:game:setup} except that $S^n$ is also known non-causally at the decoder. This communication channel is called in the \textit{Side Information (SI)} channel in the presence of a jammer. This is shown in Fig.~\ref{fig:side:information:game:setup}. 
\begin{figure}[!ht]
  \begin{center}
    \includegraphics[trim=0cm 6cm 0cm 1cm, scale=0.3]{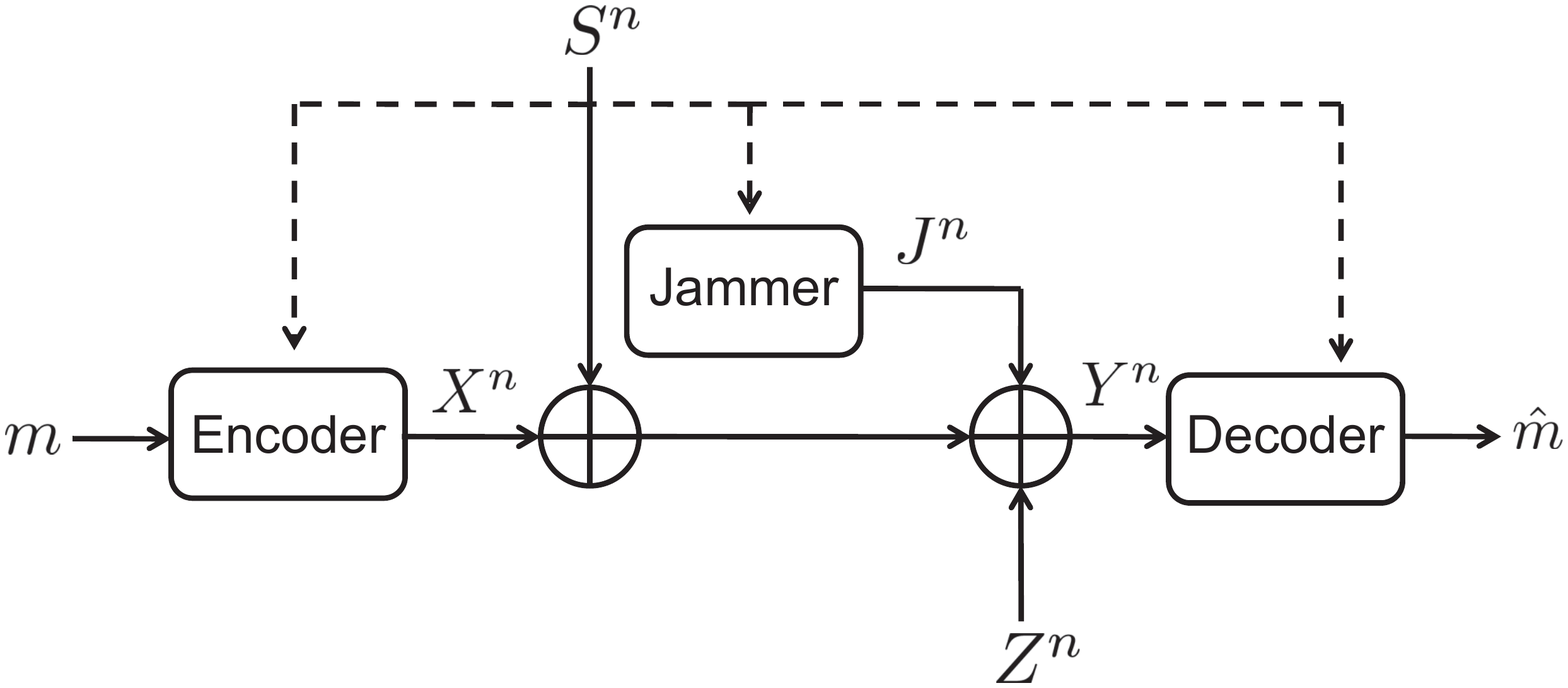}
    \caption{The Side Information Game Setup}
    \label{fig:side:information:game:setup}
  \end{center}
\end{figure}
%
Before we discuss the SI game utility, we need to review a related result on the capacity of a channel with state and no jamming adversary, with the state information available at both the encoder and the decoder. We first state the result for the discrete version of the problem.
\begin{lemma}[Chap. 7,~\cite{elgamal-kim}]\label{lem:capacity:si:discrete}
The capacity of a channel with a random state, where the state is known non-causally at both the encoder and the decoder, is 
\begin{equation}\label{eq:SI}
C=\max_{p(x|s)} I(X;Y|S).
\end{equation}
\end{lemma}
For the Gaussian version of this problem, the capacity is given in the following lemma.
\begin{lemma}[Chap. 7,~\cite{elgamal-kim}]\label{lem:capacity:si:gaussian}
The capacity of the AWGN channel with a white Gaussian state known both at the encoder and decoder is 
\begin{equation*}\label{eq:capacity:gaussian:side:information}
C_{SI}=\frac{1}{2}\log(1+\frac{P_U}{\sigma^2}).
\end{equation*}
The maximizing user channel input distribution is i.i.d. Gaussian.
\end{lemma}
Here, the subscript $SI$ stands for Side Information.
Coming back to the SI game, let us now define the game utility. Alluding to the expression in~\eqref{eq:SI}, let the SI game utility be the mutual information quantity $I(X^n;Y^n|S^n)$. 
Here, the the mutual information is evaluated for some joint distribution $p(s^n,x^n,j^n,y^n)=p(s^n)p(x^n|s^n)p(j^n|x^n)p(y^n|x^n,j^n)$. The user is the maximizing player while the jammer is the minimizing player in this zero sum game. The user chooses $p(x^n|s^n)$ given its power constraint $P_U$ so as to maximize $I(X^n;Y^n|S^n)$. The jammer chooses $p(j^n|s^n)$ under its power constraint $P_J$ so as to minimize $I(X^n;Y^n|S^n)$. Note again that $p(s^n)$ and $p(y^n|x^n,s^n,j^n)$ are fixed. 

A Nash equilibrium of the SI game is now given in the following lemma. Here, the equilibrium involves
an i.i.d. Gaussian jamming strategy. Hence, the jamming
signal and the noise together act like white Gaussian noise. The resulting
channel is the standard AWGN channel with additive state, where the state is known both to the encoder and the decoder.
The i.i.d. Gaussian coding scheme which achieves capacity, is henceforth referred to as i.i.d. Gaussian user strategy in the SI game.
%
\begin{lemma}\label{lem:si:game:saddle:point}
The Side Information game has a Nash equilibrium saddle point. The i.i.d. Gaussian user strategy and the white Gaussian jamming independent of state form a Nash equilibrium. 
\end{lemma} 
%
Before giving the proof details, we briefly discuss this result and state a corollary.
Lemma~\ref{lem:si:game:saddle:point} implies that when the user employs an i.i.d. Gaussian strategy, the best jamming strategy is memoryless linear jamming and conversely, when the jammer employs linear jamming, the best user strategy is i.i.d. Gaussian. Note that at equilibrium the jammer disregards the state knowledge completely whilst choosing its jamming signal. Using Definition~\ref{def:capacity} we have the following corollary.
\begin{corollary}\label{cor:capacity:si:game}
The capacity of an AWGN channel with an additive jammer and an additive white Gaussian state, where the state is known at both at the encoder and the decoder as well as the jammer, is 
\begin{equation}\label{eq:capacity:SI:jammer}
C_{SI-J}=\frac{1}{2} \log \left(1+\frac{P_U}{P_J+\sigma^2}\right).
\end{equation}
\end{corollary}
\begin{proof}[Proof of Lemma~\ref{lem:si:game:saddle:point}]
To prove this result, we first assume that the user signaling is i.i.d Gaussian and show that the best jamming strategy is  memoryless and linear in state. Then, we assume the jamming signal to be memoryless and linear in state and prove that the best user signaling is i.i.d Gaussian signaling. With the strategies so determined,  
it will then follow that the above pair of user and jammer strategies are the equilibrium saddle point strategies.
 
So let the user strategy be i.i.d Gaussian signaling. Recollect from previous discussion that the jammer aims to minimize $I(X^n;Y^n|S^n)$. We know that
\begin{equation}\label{eq:mutual:inf:1}
I(X^n;Y^n|S^n)=h(X^n|S^n)-h(X^n|Y^n,S^n).
\end{equation}
The jammer can only affect the second term of the RHS of~\eqref{eq:mutual:inf:1}. Thus, it will choose its signal $J^n$ so as to maximize $h(X^n|Y^n,S^n)$. Here, we have
\begin{equation}\label{eq:mutual:inf:2}
h(X^n|Y^n,S^n)=h(X^n-AY^n-BS^n|Y^n,S^n)
\end{equation}
for any constant matrices $A,B\in \mathbb{R}^{n\times n}$. Since entropy decreases upon conditioning, we have 
\begin{equation}\label{eq:mutual:inf:3}
h(X^n-AY^n-BS^n|Y^n,S^n)\leq h(X^n-AY^n-BS^n).
\end{equation}
Here, let $(AY^n-BS^n)$ be the LLSE estimate of $X^n$ from the channel output $Y^n$ and the state $S^n$ and hence, $(X^n-AY^n-BS^n)$ be the LLSE estimation error. Let the variance of this LLSE estimation error be $\Lambda_e$. 
Noting that the entropy is maximized by a Gaussian distribution for a fixed covariance matrix~\cite{gallager} we have
\begin{equation}\label{eq:mutual:inf:4}
h(X^n-AY^n-BS^n)\leq \frac{1}{2}\log( (2\pi e)^n |\Lambda_{err}| )
\end{equation} 
where, $|\Lambda_{err}|$ is the determinant of $\Lambda_{err}$.
Hence, from~\eqref{eq:mutual:inf:2},~\eqref{eq:mutual:inf:3} and~\eqref{eq:mutual:inf:4}, it is seen that
\begin{equation*}\label{eq:mutual:inf:1234}
h(X^n|Y^n,S^n)\leq \frac{1}{2}\log( (2\pi e)^n |\Lambda_{err}| )
\end{equation*} 
We now show that the best jamming strategy is Gaussian and jointly Gaussian with state $S^n$. For ease of notation given any signal $V^n=(V_1,V_2,\dots,V_n)$, we use the vector representation $\vec{V}=(V_1,V_2,\dots,V_n)^T$ whenever necessary. 
Before proving this result, we show that the best jammer strategy will be zero mean.
\begin{claim}
The best jamming strategy is a zero mean jamming strategy.
\end{claim}

{\em Proof:} 
Let $J^n$ be a feasible jamming strategy, where $\mathbb{E}[\vec{J}]=\boldsymbol{\mu}$. Let the resulting mutual information be $I(X^n;Y^n|S^n)$. Now let us define another jamming strategy $\tilde{J}^n$ such that $\vec{\tilde{J}}=\vec{J}-\boldsymbol{\mu}$, and where $\tilde{Y}^n$ is the corresponding channel output. Note that if $J^n$ is a feasible jamming signal then so is $\tilde{J}^n$ since $\mathbb{E}[\vec{\tilde{J}}^T\vec{\tilde{J}}]\leq \mathbb{E}[\vec{J}^T\vec{J}]\leq P_J$. Also note that $\mathbb{E}[\vec{\tilde{J}}]=\vec{0}$, where $\vec{0}$ is the zero vector. Here we have,
\begin{subequations}
\begin{IEEEeqnarray*}{rCl}
I(X^n;Y^n|S^n)&=&h(Y^n|S^n)-h(Y^n|X^n,S^n)\\
              &\stackrel{(a)}=&h(Y^n-\boldsymbol{\mu})-h(Y^n-\boldsymbol{\mu}|X^n,S^n)\\
							&=&h(\tilde{Y}^n)-h(\tilde{Y}^n|X^n,S^n)\\
							&\stackrel{(b)}=&I(X^n;\tilde{Y}^n|S^n).
\end{IEEEeqnarray*}
\end{subequations}
\addtocounter{equation}{-1}
Here, $(a)$ follows by noting that entropy is invariant under translation~\cite{gallager}. $(b)$ follows from the definition of mutual information.
Thus, for every feasible jamming strategy there exits an equivalent zero-mean jamming strategy which results in the same mutual information value.~\IEEEQEDopen

Thus, from here on all the jamming strategies considered are zero mean strategies. 
\begin{claim}
A linear jamming strategy, jointly Gaussian with state, is the best jamming strategy.
\end{claim}
{\em Proof:}
Let $J^n$ be the best choice for the jammer given an i.i.d Gaussian user strategy which minimizes $h(X^n|Y^n,S^n)$ and let $\Lambda_{err}$ be the corresponding error covariance matrix. Since the error is the LLSE estimator error, $\Lambda_{err}$ is given as~\cite{gallager}
\begin{equation*}\label{eq:LLSE:error:covariance}
\Lambda_{err}=\Lambda_{\vec{X}\vec{X}}-\Lambda_{\vec{X},\vec{Y}}\Lambda_{\vec{Y},\vec{Y}}^{-1} \Lambda_{\vec{X},\vec{Y}}^T.
\end{equation*}
Here, $\Lambda_{\vec{X}\vec{X}}$ and $\Lambda_{\vec{Y}\vec{Y}}$ are the autocovariance matrices of $\vec{X}$, $\vec{Y}$ respectively and $\Lambda_{\vec{X}\vec{Y}}$ the covariance matrix of $\vec{X}$ and $\vec{Y}$ .
Let us define a random variable $\vec{R}=R^n$ as 
\begin{equation*} 	
\vec{R}= \vec{J}-\frac{\mathbb{E}[\vec{J}^T\vec{S}]}{P_S} \vec{S}.
\end{equation*}
Observe that $\vec{R}$ is uncorrelated to $\vec{S}$ since $\mathbb{E}[\vec{R}^T\vec{S}]=\vec{0}=\mathbb{E}[\vec{R}^T]\mathbb{E}[\vec{S}]$ and $\mathbb{E}[\vec{S}]=\vec{0}$. Here, $\vec{0}$ is the $n$-length zero vector. 
Also note that $\vec{R}$ is uncorrelated to $\vec{X}$ since
\begin{subequations}\label{eq:X:uncorrelated:R}
\begin{IEEEeqnarray}{rCl}
\mathbb{E}[\vec{R}^T\vec{X}]&=& \mathbb{E}_{\vec{S}}[\mathbb{E} [\vec{R}^T\vec{X}|\vec{S}]]\\
                            &\stackrel{(a)}=& \mathbb{E}_{\vec{S}}[\mathbb{E} [\vec{R}^T|\vec{S}]\mathbb{E} [\vec{X}|\vec{S}]\\
														&\stackrel{(b)}=& \mathbb{E}_{\vec{S}}[\mathbb{E} [\vec{R}^T]\mathbb{E} [\vec{X}|\vec{S}]\\
                            &=& \mathbb{E} [\vec{R}^T]\mathbb{E}_{\vec{S}}[\mathbb{E} [\vec{X}|\vec{S}]\\
														&=& \mathbb{E} [\vec{R}^T]\mathbb{E}[\vec{X}]
\end{IEEEeqnarray}
\end{subequations}
Here, $(a)$ follows since given $\vec{S}$, $\vec{X}$ and $\vec{R}$ are uncorrelated. $\vec{R}$ is uncorrelated to $\vec{S}$ and hence, we have $(b)$. The rest then follows in a straightforward manner.

Now, the average power of the jamming signal $\vec{J}=\frac{\mathbb{E}[\vec{J}^T\vec{S}]}{P_S} \vec{S}+\vec{R}$ is
\begin{subequations}
\begin{IEEEeqnarray*}{rCl} 	
\mathbb{E}[\vec{J}^T\vec{J}]&=& \mathbb{E}[\vec{R}^T\vec{R}]+\frac{\mathbb{E}[\vec{S}^T\vec{S}]}{{P_S}^2} \mathbb{E}^2[\vec{J}^T\vec{S}]\\
&=&\mathbb{E}[\vec{R}^T\vec{R}]+\frac{\mathbb{E}^2[\vec{J}^T\vec{S}]}{{P_S}} 
\end{IEEEeqnarray*}
\end{subequations}
\addtocounter{equation}{-1}
Here, we have used the facts that $\vec{R}$ is uncorrelated to $\vec{S}$, $\mathbb{E}[\vec{S}]=\vec{0}$ and $\mathbb{E}[\vec{S}^T\vec{S}]=P_S$. Also, recall that feasibility requires that  $\mathbb{E}[\vec{J}^T\vec{J}]\leq P_J$. 

We now define a linear jammer whose jamming signal $\vec{J'}=J'^n$ is given as
\begin{equation*}\label{eq:linear:jammer} 	
\vec{J'}=\frac{\mathbb{E}[\vec{J}^T\vec{S}]}{P_S} \vec{S}+\vec{R'}.
\end{equation*}
Here, $\vec{R'}=(R'_1,R'_2,\dots,R'_n)$ and $\vec{R'}\sim \mathcal{N}\left(\mathbb{E}[\vec{R}], K_{\vec{R}\vec{R}} \right)$. Since $\mathbb{E}[\vec{J}]=0$ and $\mathbb{E}[\vec{S}]=0$, it follows that $\mathbb{E}[\vec{R}]=0$. Also, $R'^n$ is independent of state $S^n$.  
It can be directly seen that $\vec{J}'$ and $\vec{J}$ are such that
\begin{equation*}
\mathbb{E}[\vec{J}'^T\vec{J}']=\mathbb{E}[\vec{J}^T\vec{J}]
\end{equation*}
Thus, the linear jammer with jamming signal $J'^n$ has the same average power as that of the best jammer $J^n$ and hence, is a feasible jamming signal.

Now given the communication setup and a fixed input strategy, it is seen that if 
\begin{subequations}\label{eq:XY:covariance:equality}
\begin{IEEEeqnarray}{rCl}
K_{\vec{X}\vec{Y'}}=K_{\vec{X}\vec{Y}}\\
K_{\vec{Y}'\vec{Y}'}=K_{\vec{Y}\vec{Y}}
\end{IEEEeqnarray}
\end{subequations}
where $\vec{Y}$ and $\vec{Y'}$ are the channel outputs when the jammer strategy is $\vec{J}$ and $\vec{J}'$ respectively, then $\Lambda_{{err}'}=\Lambda_{err}$. Here $\Lambda_{{err}'}$ is the error covariance matrix when the jammer is linear and is given as
\begin{equation*}\label{eq:LLSE:error':covariance}
\Lambda_{{err}'}=\Lambda_{\vec{X}\vec{X}}-\Lambda_{\vec{X},\vec{Y}'}\Lambda_{\vec{Y}',\vec{Y}'}^{-1} \Lambda_{\vec{X},\vec{Y}'}^T.
\end{equation*}
$\Lambda_{\vec{Y}'\vec{Y}'}$ is the autocovariance matrix of the channel output $\vec{Y}'$ and $\Lambda_{\vec{X}\vec{Y}'}$ is the covariance matrix of $\vec{X}$ and $\vec{Y}'$. Noting that $\mathbb{E}[\vec{J}]=\vec{0}=\mathbb{E}[\vec{J'}]$, it is seen that~\eqref{eq:XY:covariance:equality} follows if 
\begin{subequations}\label{eq:XJ:covariance:equality}
\begin{IEEEeqnarray*}{rCl}
\mathbb{E}[\vec{X}\vec{J'}^T]=\mathbb{E}[\vec{X}\vec{J}^T]\\
\mathbb{E}[\vec{J}'\vec{J'}^T]=\mathbb{E}[\vec{J}\vec{J}^T].
\end{IEEEeqnarray*}
\end{subequations}
\addtocounter{equation}{-1}
Now,
\begin{subequations}\label{eq:XJ:XJ'}
\begin{IEEEeqnarray}{rCl}
\mathbb{E}[\vec{X}\vec{J}^T]&=&\mathbb{E}\left[\vec{X}\left(\frac{\mathbb{E}[\vec{J}^T\vec{S}]}{P_S}\vec{S}+\vec{R}\right)^T\right]\nonumber\\
&=&\frac{\mathbb{E}[\vec{J}^T\vec{S}]}{P_S} \mathbb{E}[\vec{X}\vec{S}^T]+\mathbb{E}[\vec{X}\vec{R}^T]\nonumber\\
&\stackrel{(a)}=&\frac{\mathbb{E}[\vec{J}^T\vec{S}]}{P_S} \mathbb{E}[\vec{X}\vec{S}^T]+\mathbb{E}[\vec{X}]\mathbb{E}[\vec{R}^T]\\
&\stackrel{(b)}=&\frac{\mathbb{E}[\vec{J}^T\vec{S}]}{P_S} \mathbb{E}[\vec{X}\vec{S}^T]+\mathbb{E}[\vec{X}]\mathbb{E}[\vec{R}'^T]\\
&=&\mathbb{E}[\vec{X}\vec{J}'^T]\nonumber
\end{IEEEeqnarray}
\end{subequations}
\addtocounter{equation}{-1}	
Here, $(a)$ follows from~\eqref{eq:X:uncorrelated:R}, and $(b)$ from noting that $\mathbb{E}[\vec{R}]=\mathbb{E}[\vec{R}']$.
Similarly, we have
%
\begin{subequations}\label{eq:JJ:equal:J'J'}
\begin{IEEEeqnarray}{rCl}
\mathbb{E}[\vec{J}\vec{J}^T]&=&\mathbb{E}\left[\left(\frac{\mathbb{E}[\vec{J}^T\vec{S}]}{P_S}\vec{S}+\vec{R}\right)\left(\frac{\mathbb{E}[\vec{J}^T\vec{S}]}{P_S}\vec{S}+\vec{R}\right)^T\right]\nonumber\\
&=&\frac{\mathbb{E}^2[\vec{J}^T\vec{S}]}{P^2_S}\mathbb{E}[\vec{S}\vec{S}^T]+\frac{\mathbb{E}[\vec{J}^T\vec{S}]}{P_S} \left(\mathbb{E}[\vec{R}\vec{S}^T]+\mathbb{E}[\vec{S}\vec{R}^T]\right)+\mathbb{E}[\vec{R}\vec{R}^T] \nonumber\\
&\stackrel{(a)}=&\frac{\mathbb{E}^2[\vec{J}^T\vec{S}]}{P_S}\mathbb{E}[\vec{S}\vec{S}^T]+\frac{\mathbb{E}[\vec{J}^T\vec{S}]}{P_S}\left(\mathbb{E}[\vec{R}]\mathbb{E}[\vec{S}^T]+\mathbb{E}[\vec{S}]\mathbb{E}[\vec{R}^T]\right)+\mathbb{E}[\vec{R}\vec{R}^T] \\
&\stackrel{(b)}=&\frac{\mathbb{E}^2[\vec{J}^T\vec{S}]}{P_S}\mathbb{E}[\vec{S}\vec{S}^T]+\frac{\mathbb{E}[\vec{J}^T\vec{S}]}{P_S}\left(\mathbb{E}[\vec{R}']\mathbb{E}[\vec{S}^T]+\mathbb{E}[\vec{S}]\mathbb{E}[\vec{R}'^T]\right)+\mathbb{E}[\vec{R}'\vec{R}'^T] \\
&=&\mathbb{E}\left[\left(\frac{\mathbb{E}[\vec{J}^T\vec{S}']}{P_S}\vec{S}+\vec{R}'\right)\left(\frac{\mathbb{E}[\vec{J}^T\vec{S}']}{P_S}\vec{S}+\vec{R}'\right)^T\right]\nonumber\\
&=&\mathbb{E}[\vec{J}'\vec{J}'^T]\nonumber
\end{IEEEeqnarray}
\end{subequations}
%
Here, we see that $(a)$ is a result of $\vec{R}$ being uncorrelated to $\vec{S}$. In the next step, $(b)$ is true since $\vec{R}$ and $\vec{R}'$ have the same mean and covariance matrix. We thus have, $\Lambda_{{err}'}=\Lambda_{err}$. 

Finally, since $J'^n$ is a Gaussian signal which is jointly Gaussian with $S^n$, the resulting LLSE estimator is the MMSE estimator and the corresponding error is a Gaussian vector. Thus, we have the following

%
\begin{IEEEeqnarray*}{rCl}\label{eq:J:J':conditional:entropy}
h(X^n-AY^n-BS^n|Y^n,S^n)&\leq& \frac{1}{2}\log \left( (2\pi e)^n|\Lambda_{{err}}|\right) \\
                &= & \frac{1}{2}\log \left( (2\pi e)^n|\Lambda_{{err}'}|\right) \\    
                &= & h(X^n-AY'^n-BS^n)\\     
								&= & h(X^n-AY'^n-BS^n|Y'^n,S^n) \\
								&= & h(X^n|Y'^n,S^n)\yesnumber
\end{IEEEeqnarray*}
%
Hence, from~\eqref{eq:mutual:inf:1} and~\eqref{eq:J:J':conditional:entropy} we have 
\begin{equation*}
I(X^n;Y^n|S^n)\geq I(X^n;Y'^n|S^n).
\end{equation*}
This implies that the jointly Gaussian linear jamming is the best jamming strategy.~\IEEEQEDopen

We now determine the nature of the this jamming signal. From now on, let $J^n$ be the best Gaussian and jointly Gaussian (with $S^n$) strategy. The feasibility requires that $\mathbb{E}[\vec{J}^T\vec{J}]\leq P_J$. 
\begin{claim}
The i.i.d Gaussian jamming strategy, linear in state, is the best jamming strategy.
\end{claim}
%
{\em Proof:}
Consider $h(X^n|Y^n,S^n)$ again. From the entropy chain rule~\cite{gallager}, we know that 
\begin{equation}\label{eq:mutual:inf:5:1}
h(X^n|Y^n,S^n)= \sum_{i=1}^n h(X_i|Y^n,S^n,X_1^{i-1})
\end{equation}
where, $X_{1}^{i-1}=(X_1,X_2,\dots,X_{i-1})$. Now since the user strategy is i.i.d Gaussian, we have
\begin{equation}\label{eq:mutual:inf:5:2}
\sum_{i=1}^n h(X_i|Y^n,S^n,X_1^{i-1})= \sum_{i=1}^n h(X_i|Y^n,S^n).
\end{equation}
Noting that conditioning reduces entropy it follows that
\begin{equation}\label{eq:mutual:inf:5:3}
\sum_{i=1}^n h(X_i|Y^n,S^n)\leq \sum_{i=1}^n h(X_i|Y_i,S_i).
\end{equation}
Thus, from~\eqref{eq:mutual:inf:5:1},~\eqref{eq:mutual:inf:5:2} and~\eqref{eq:mutual:inf:5:3}  we have
\begin{equation}\label{eq:mutual:inf:5:123}
h(X^n|Y^n,S^n)\leq \sum_{i=1}^n h(X_i|Y_i,S_i).
\end{equation}
Now the mutual information $I(X^n;Y^n|S^n)$, when the jamming strategy $J^n$ is jointly Gaussian and linear in $S^n$, is
\begin{subequations}\label{eq:mutual:inf:6}
\begin{IEEEeqnarray}{rCl}
I(X^n;Y^n|S^n)&=&h(X^n)-h(X^n|Y^n,S^n)\nonumber\\
              &\stackrel{(a)}=& \left(\sum_{i=1}^n h(X_i)\right) -h(X^n|Y^n,S^n)\\
							&\stackrel{(b)}\geq& \sum_{i=1}^n \left(h(X_i) -h(X_i|Y_i,S_i)\right)\\
              &=& \sum_{i=1}^n I(X_i;Y_i|S_i)\nonumber\\
							&\stackrel{(c)}=& \sum_{i=1}^n I(X_i;Y'_i|S_i)\\
							&=&  I(X^n;Y'^n|S^n)\nonumber
\end{IEEEeqnarray}
\end{subequations}
Here, $(a)$ follows from the fact that the user strategy is i.i.d Gaussian. $(b)$ is a direct consequence of~\eqref{eq:mutual:inf:5:123}. Now, for $(c)$ we have the following. We know that $(X^n,S^n,Z^n,J^n,R^n,Y^n)$ is jointly Gaussian. Let us define $F^n=(X^n,S^n,Z^n,J^n,R^n,Y^n)$ where, $F_i=(X_i,S_i,Z_i,J_i,R_i,Y_i)$. Since $\vec{R}$ is not a memoryless vector, it follows that $\vec{J}$ and $\vec{Y}$ are not memoryless vectors. This also results in $\vec{F}$ not being memoryless. Now let us define a Gaussian vector $\vec{R}'=(R'_1,R'_2,\dots,R'_n)$, where $R'_i\sim \mathcal{N}(0,\sigma_{R_i}^2)$, $i=1,2,\dots,n$. Next, let us define a new linear jamming strategy $J'^n=(J'_1,J'_2,\dots,J'_n)$, $i=1,2,\dots,n$, where $J'_i=\beta S_i+R'_i$, $i=1,2,\dots,n$. Note that $\vec{J}'$ is memoryless. Now since $\vec{R}'$ and $\vec{J}'$ are memoryless, it follows that $\vec{Y}'=(Y_1,Y_2,\dots, Y_n)$, where $Y_i=X_i+S_i+Z_i+J_i$ and $\vec{F}'=(F'_1,F'_2,\dots, F'_n)$, where $F'_i=(X_i,S_i,Z_i,J'_i,R'_i,Y'_i)$ are also memoryless vectors. Observe that $J'^n$ is a feasible jamming strategy as $\mathbb{E}[\vec{J}'^T\vec{J}']=\mathbb{E}[\vec{J}^T\vec{J}]\leq P_J$. Thus, by choosing the jamming strategy in the above manner we have $I(X_i;Y_i|S_i)=I(X_i;Y'_i|S_i)$, and hence, equality in $(c)$. 

Finally, the $\vec{R}'$ can be made i.i.d as the logarithm function in $I(X;Y|S)$ is a concave function. Hence, the best linear strategy is given as $J'_i=\beta_i S_i+R'_i$, $i=1,2,\dots,n$ where $\beta$ is a constant and $R'_i\sim\mathcal{N}(0,\sigma_{R'}^2)$. Note that the average power constraint dictates that $\beta^2\sigma_S^2+\sigma_R'^2\leq P_J$.~\IEEEQEDopen
%

Before we determine the optimal values of $(\beta,\sigma_R^2)$, we determine the best user strategy when the jamming strategy is i.i.d Gaussian and linear in state. 
\begin{claim}
If the jammer strategy is linear in state and i.i.d, the user employs i.i.d Gaussian signaling.
\end{claim}
%
{\em Proof:}
Let us assume that the jammer chooses an i.i.d signal $J^n=(J_1,J_2,\dots,J_n)$, where $J_i=\beta S_i+R_i$.
Then, we have the scenario depicted in Fig.~\ref{fig:si:minimax:equivalent}. 
\begin{figure}[!ht]
  \begin{center}
    \includegraphics[trim=0cm 3cm 0cm 0cm,scale=0.3]{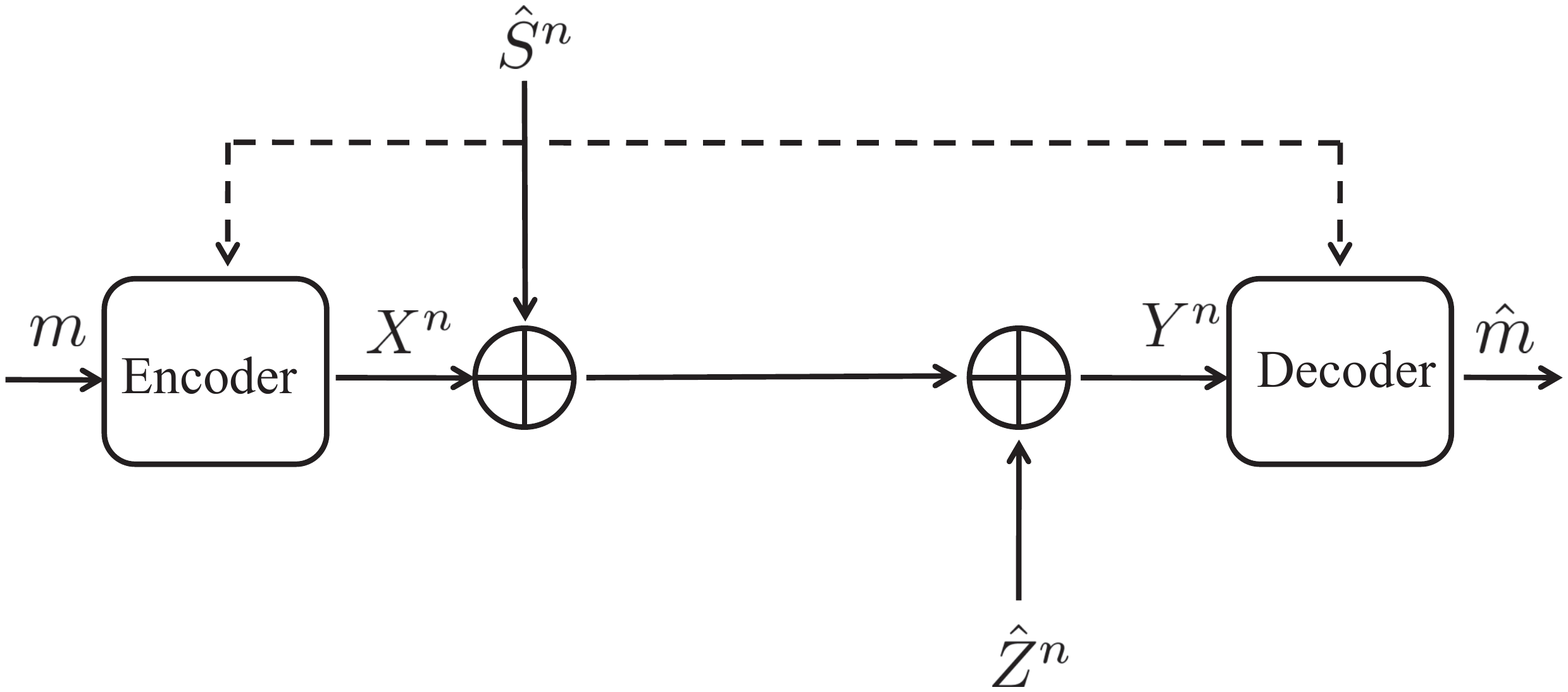}
    \caption{The SI Game Equivalent Setup When Jamming is Memoryless and Linear in State}
    \label{fig:si:minimax:equivalent}
  \end{center}
\end{figure}
Here, $\hat{S}^n=(\hat{S}_1,\hat{S}_2,\dots,\hat{S}_n)$ is the effective white Gaussian state of the channel and $\hat{Z}^n=(\hat{Z}_1,\hat{Z}_2,\dots,\hat{Z}_n)$ is the effective independent noise in the channel. Note that $\hat{S}_i=(1+\beta)S_i$ and $\hat{Z}_i=Z_i+R_i$ for $i=1,2,\dots,n$. This is an instance of a channel with state $\hat{S}^n$, where $\hat{S}^n$ is known to the encoder and the decoder. From~\cite{elgamal-kim} and Lemma~\ref{lem:si:game:saddle:point}, we know that the capacity achieving user signaling for this channel is i.i.d Gaussian signaling.~\IEEEQEDopen

We have thus determined the saddle point user and jammer strategies. Given the the saddle point jamming strategy $J^n=(J_1,J_2,\dots,J_n)$, where $J_i=\beta S_i+R_i$ the jammer chooses $(\beta,\sigma_R^2)$ so as to induce the worst channel for the user. However, the jammer is average power constrained and hence, $\beta^2\sigma_S^2+\sigma_R^2\leq P_J$. We now state  the optimal value of $(\beta,\sigma_R^2)$.
\begin{claim}
The saddle point jamming strategy is white Gaussian noise jamming.
\end{claim}
%
{\em Proof:}
Since the jammer behaviour at equilibrium is memoryless linear (in state) Gaussian jamming, note that from the jammer's perspective any component dependent on the state will be known at the encoder-decoder in the SI game which  can be completely canceled by the decoder. Thus, to maximize its adversarial impact on the mutual information, the jammer will choose $J^n$ to be completely independent of $S^n$ and hence, $\beta=0$ and $\sigma_R^2=P_J$. The jammer thus effectively acts as an independent white Gaussian noise signal.~\IEEEQEDopen
%

It follows that the equilibrium saddle point value, and thence the capacity of the channel, is given as in~\eqref{eq:capacity:SI:jammer}. This completes the proof of Lemma~\ref{lem:si:game:saddle:point}.
\end{proof}
We have seen that at equilibrium the jammer, even in the presence of state knowledge, disregards it completely whilst choosing its jamming signal. 
\subsection{The Costa Game}
%
We first show that the Costa game utility is upper bounded by the SI game utility. Recall that the player utility in the Costa game is $I(U^n;Y^n)-I(U^n;S^s)$, where $U^n$ is an appropriate auxiliary variable, while the SI game utility is $I(X^n;Y^n|S^n)$. We then have the following result.
\begin{lemma}
$I(U^n;Y^n)-I(U^n;S^n) \leq I(X^n;Y^n|S^n)$, where $U^n-(X^n,S^n)-Y^n$.
\end{lemma}
 
\indent\textit{Proof:}
Using elementary entropy properties, we have
\begin{equation}\label{eq:costa:mutual:inf:1}
I(U^n;Y^n)-I(U^n;S^n)=h(U^n|S^n)-h(U^n|Y^n)
\end{equation}
Since entropy is decreased upon conditioning, it is seen that
\begin{eqnarray*}\label{eq:costa:mutual:inf:2}
h(U^n|S^n)-h(U^n|Y^n)&\leq& h(U^n|S^n)-h(U^n|Y^n,S^n)\\
                &=& I(U^n;Y^n|S^n)
\end{eqnarray*}
Thus, it follows that
\begin{equation}\label{eq:costa:mutual:inf:3}
I(U^n;Y^n)-I(U^n;S^n)\leq I(U^n;Y^n|S^n).
\end{equation}
From the Markov chain $U^n-(X^n,S^n)-Y^n$ it is seen that
\begin{equation}\label{eq:costa:mutual:inf:4}
I(U^n;Y^n|S^n)\leq I(X^n;Y^n|S^n).
\end{equation}
From~\eqref{eq:costa:mutual:inf:3} and~\eqref{eq:costa:mutual:inf:4}, we finally have  
\begin{equation*}\label{eq:costa:leq:si}
I(U^n;Y^n)-I(U^n;S^n)\leq I(X^n;Y^n|S^n). ~\QEDopen
\end{equation*}
%
%
Therefore, for any pair of user and jammer strategies the Costa game utility is upper bounded by the SI game utility. We will come back to this result after we prove that the Costa game has a Nash equilibrium saddle point. We now prove Theorem~\ref{thm:costa:game:saddle:point}.\\
\\
\indent \textit{Proof of Theorem~\ref{thm:costa:game:saddle:point}}\\
The proof entails, firstly, the verification of the saddle point strategies given in Theorem~\ref{thm:costa:game:saddle:point}. Then, it is shown that the utility at this saddle point is as given in~\eqref{eq:costa:game:saddle:point}. 
We begin with the following lemma.
\begin{lemma}\label{lem:u:DPC:j:iid}
When the user strategy is DPC strategy, the best jamming strategy is white Gaussian jamming independent of state.
\end{lemma}
\indent\textit{Proof:}
For ease of notation given any signal $V^n=(V_1,V_2,\dots,V_n)$, we use the vector representation $\vec{V}=(V_1,V_2,\dots,V_n)^T$ whenever necessary. To begin with, let us fix the user strategy as the dirty paper coding strategy with $\alpha=P_U/(P_U+P_J+\sigma^2)$. Note that this user strategy is an i.i.d. Gaussian strategy. Recall that the jammer chooses a signal $J^n$ with the aim to minimize $I(U^n;Y^n)-I(U^n;S^n)$ in the Costa game. Since $J^n$ is a feasible signal, we have $\mathbb{E}[\vec{J}^T\vec{J}]\leq P_J$. Referring to~\eqref{eq:costa:mutual:inf:1}, it is seen that the jammer can only affect the term $h(U^n|Y^n)$ so as to minimize $I(U^n;Y^n)-I(U^n;S^n)$. Now for any $J^n$, we have
\begin{subequations}\label{eq:jammer:in:costa:1}
\begin{IEEEeqnarray}{rCl}
h(U^n|Y^n)&\stackrel{(a)}=&h(X^n+\alpha S^n|Y^n)\\
          &=&h(X^n+\alpha S^n-\alpha Y^n|Y^n)\nonumber\\  
					&=&h((1-\alpha)X^n-\alpha J^n-\alpha Z^n|Y^n)\nonumber\\  
					&\stackrel{(b)}\leq& h((1-\alpha)X^n-\alpha J^n-\alpha Z^n)\\
					&\stackrel{(c)}\leq& \frac{1}{2}\log\left( (2\pi e)^n |\Lambda_{W}| \right)	.
\end{IEEEeqnarray} 
\end{subequations}
Here, $W^n=((1-\alpha)X^n-\alpha J^n-\alpha Z^n)$. Also, $\Lambda_{W}$ is the covariance matrix of $W^n$ and  $|\Lambda_{W}|$ its determinant. Note that $(a)$ follows from the fact that user employs a dirty paper coding scheme. Conditioning reduces entropy and hence, we have $(b)$. Finally, $(c)$ is true from the well known fact that for a given covariance matrix, the Gaussian distribution maximizes entropy~\cite{elgamal-kim}.

Now, let us define another Gaussian jamming strategy $J'^n=(J'_1,J'_2,\dots,J'_n)$, where $J'^n\sim\mathcal{N}(0,P_J\vec{I}_n)$ and $J'^n$ is independent of the state $S^n$. Here, $\vec{I}_n\in\mathbb{R}^{n\times n}$ is the identity matrix. Note that $J'^n$ is a feasible jamming strategy since $\mathbb{E}[\vec{J'}^T\vec{J'}]=P_J\leq P_J$. For the jamming strategy $J'^n$, let $Y'^n$ be the corresponding channel output where
\begin{equation}\label{eq:Y':defn}
Y'^n=X^n+S^n+J'^n+Z^n.
\end{equation}
Also, define $W'^n$ as follows.
\begin{equation}\label{eq:W':definition}
W'^n=((1-\alpha)X^n-\alpha J'^n-\alpha Z^n)
\end{equation}
Let $\Lambda_{W'}$ be the covariance matrix of $W'^n$ and let $|\Lambda_{W'}|$ be its determinant. 
Then, we have the following.
\begin{subequations}\label{eq:lambda:W:expression}
\begin{IEEEeqnarray}{rCl}
\Lambda_{W}&=&(1-\alpha^2)\Lambda_X+\alpha^2 \Lambda_J+\alpha^2 \Lambda_Z\\
\Lambda_{W'}&=&(1-\alpha^2)\Lambda_X+\alpha^2 \Lambda_{J'}+\alpha^2 \Lambda_Z.
\end{IEEEeqnarray}
\end{subequations}
We now make a series of claims.
\begin{claim}\label{claim:lambda:W:lambda:W'}
$|\Lambda_{W}|\leq|\Lambda_{W'}|$.
\end{claim}
%
\indent\textit{Proof:}
Consider the following. 
\begin{subequations}\label{eq:lambda:J:lambda:J'}
\begin{IEEEeqnarray}{rCl}
|\Lambda_{W}|&\stackrel{(a)}\leq& \prod_{i=1}^n \left((1-\alpha)^2P_U+\alpha^2\mathbb{E}[J_i^2]+\alpha^2 \sigma^2\right)\\
             &=& \left(\left[\prod_{i=1}^n \left((1-\alpha)^2P_U+\alpha^2\mathbb{E}[J_i^2]+\alpha^2 \sigma^2\right)\right]^{\frac{1}{n}}\right)^n\nonumber\\
             &\stackrel{(b)}\leq& \left(\frac{\sum_{i=1}^n \left((1-\alpha)^2P_U+\alpha^2\mathbb{E}[J_i^2]+\alpha^2 \sigma^2\right)}{n}\right)^n\\
						&=& \left((1-\alpha)^2P_U+\alpha^2\frac{\sum_{i=1}^n \mathbb{E}[J_i^2]}{n}+\alpha^2 \sigma^2\right)^n\nonumber\\
						&\stackrel{(c)}\leq& \left((1-\alpha)^2P_U+\alpha^2 P_J+\alpha^2 \sigma^2\right)^n\\
						&=& |(1-\alpha)^2 P_U \vec{I}_n+\alpha^2P_J\vec{I}_n+\alpha^2 \sigma^2\vec{I}_n|\nonumber\\
						&=&|(1-\alpha^2)\Lambda_X+\alpha^2 \Lambda_{J'}+\alpha^2 \Lambda_Z|\nonumber\\
						&\stackrel{(d)}=&|\Lambda_{W'}|
\end{IEEEeqnarray}
\end{subequations}
Here, $(a)$ follows from Hadamard's inequality~\cite{elgamal-kim}. 
 Since the Geometric Mean (GM) is less than the Arithmetic Mean (AM), we have $(b)$. The feasibility of the jammer strategy requires $\sum_{i=1}^n\mathbb{E}[J_i^2]\leq n P_J$ and hence, $(c)$ is true. Finally, $(d)$ follows from~\eqref{eq:lambda:W:expression}. ~\IEEEQEDopen\
%
\begin{claim}\label{claim:W':uncorrelated:Y'}
$W'^n$ is independent of $Y'^n$.
\end{claim}
\indent\textit{Proof:}
Since $W'^n$ and $Y'^n$ are Gaussian, it is sufficient to show that $W'^n$ is uncorrelated to $Y'^n$ i.e. $\mathbb{E}[\vec{W'}\vec{Y'}^T]=\mathbb{E}[\vec{W'}]\mathbb{E}[\vec{Y'}^T]$. We now proceed to establish this fact. 
\begin{subequations}\label{eq:E:W:E:Y}
\begin{IEEEeqnarray*}{rCl}
\mathbb{E}[\vec{W'}]\mathbb{E}[\vec{Y'}^T]&=&\mathbb{E}[(1-\alpha)\vec{X}-\alpha \vec{J'}-\alpha \vec{Z}]\mathbb{E}[\vec{Y'}^T]\\
                                        &=& \vec{0}.\mathbb{E}[\vec{Y'}^T]=\vec{0}.
\end{IEEEeqnarray*}
\end{subequations}
\addtocounter{equation}{-1}
Similarly,
\begin{subequations}\label{eq:E:WY}
\begin{IEEEeqnarray}{rCl}
\mathbb{E}[\vec{W'}\vec{Y'}^T]&=&\mathbb{E}[((1-\alpha)\vec{X}-\alpha \vec{J'}-\alpha \vec{Z})\vec{Y'}]\nonumber\\
                          &=&\mathbb{E}[(1-\alpha)\vec{X}\vec{Y'}^T-\alpha \vec{J'}\vec{Y'}^T-\alpha \vec{Z}\vec{Y'}^T]\nonumber\\
													&\stackrel{(a)}=&\mathbb{E}[(1-\alpha)\vec{X}\vec{X}^T-\alpha\vec{J'}\vec{J'}^T-\alpha \vec{Z}\vec{Z}^T]\\
													&=&(1-\alpha)P_U\vec{I}_n-\alpha P_J\vec{I}_n-\alpha \sigma^2\vec{I}_n\nonumber\\
													&\stackrel{(b)}=&\vec{0}.
\end{IEEEeqnarray}
\end{subequations}
\addtocounter{equation}{-1}
Here, $(a)$ follows from noting that $X^n$, $S^n$, $J'^n$ and $Z^n$ are independently chosen, and in addition, since $\mathbb{E}[\vec{S}]=\vec{0}$, $\mathbb{E}[\vec{Z}]=\vec{0}$. We also need $\mathbb{E}[\vec{X}]=\vec{0}$ and this follows from the fact that the user strategy is DPC. The equality in $(b)$ follows from using $\alpha=P_U/(P_U+P_J+\sigma^2)$ for the DPC. ~\IEEEQEDopen

We now complete the proof of Lemma~\ref{lem:u:DPC:j:iid} by showing that $J'^n$ is the best feasible jamming signal. Here
%
\begin{subequations}\label{eq:jammer:in:costa:2}
\begin{IEEEeqnarray}{rCl}
h(U^n|Y^n)&\stackrel{(a)}\leq& \frac{1}{2}\log \left( (2\pi e)^n|\Lambda_{{W}}|\right) \\
                &\stackrel{(b)}\leq & \frac{1}{2}\log \left( (2\pi e)^n|\Lambda_{{W}'}|\right) \\    
                &\stackrel{(c)}= & h((1-\alpha)X^n-\alpha J'^n-\alpha Z^n)\\     
								&\stackrel{(d)}= & h((1-\alpha)X^n-\alpha J'^n-\alpha Z^n|Y'^n)\\     
								&= & h((1-\alpha)X^n-\alpha J'^n-\alpha Z^n+\alpha Y'^n|Y'^n) \nonumber\\
								&\stackrel{(e)}= & h(U^n|Y'^n).
\end{IEEEeqnarray}
\end{subequations}
\addtocounter{equation}{-1}
Now, $(a)$ follows from~\eqref{eq:jammer:in:costa:1} while $(b)$ follows from Claim~\ref{claim:lambda:W:lambda:W'}. 
Using the definition of $W'^n$ in~\eqref{eq:W':definition}, $(c)$ is true. Claim~\ref{claim:W':uncorrelated:Y'} results in the equality in $(d)$. Finally, $(e)$ follows from noting that $U^n=X^n+\alpha S^n$ and~\eqref{eq:Y':defn}.~\IEEEQEDclosed

We discuss the equilibrium saddle point value after determining the best user strategy for an i.i.d. jamming strategy.
\begin{lemma}
When the jamming strategy is i.i.d. Gaussian, the best user strategy is the the dirty paper coding strategy.
\end{lemma}
\balance
\indent\textit{Proof:}
Let us assume that the jammer chooses an i.i.d. signal $J^n=(J_1,J_2,\dots,J_n)$, where $J_i\sim\mathcal{N}(0,P_J)$.
Then, we have the scenario depicted in Fig.~\ref{fig:jamming:iid:fixed}. 
\begin{figure}[!ht]
  \begin{center}
    \includegraphics[trim=0cm 10cm 0cm 0cm,scale=0.3]{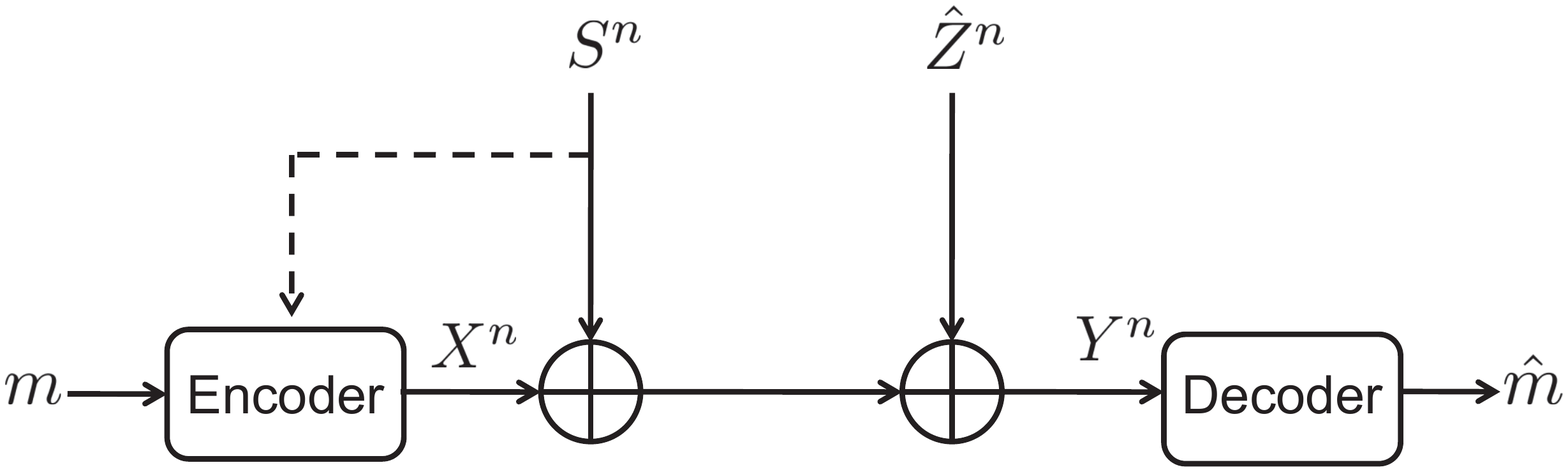}
    \caption{The Costa Game Equivalent Setup when Jamming is i.i.d. and independent of State}
    \label{fig:jamming:iid:fixed}
  \end{center}
\end{figure}
Here, $\hat{Z}^n=(\hat{Z}_1,\hat{Z}_2,\dots,\hat{Z}_n)$ is the effective independent noise in the channel, where $\hat{Z}_i=Z_i+J_i$ for $i=1,2,\dots,n$. This is an instance of a channel with state $S^n$ and noise $\hat{Z}^n$, where $S^n$ is known to the encoder . From~\cite{costa}, we know that the dirty paper code with $\alpha=P_U/(P_U+P_J+\sigma^2)$ achieves the capacity for this channel. This coding scheme is an i.i.d. Gaussian scheme. ~\IEEEQEDclosed
%

Note that the above choice of saddle point strategies results in the saddle point value equal to $C_{C-J}$ given in~\eqref{eq:costa:game:saddle:point}. This follows from the fact that this is an instance of a Costa channel where the average user	 power constraint is $P_U$ and the effective AWGN noise in the channel has variance $\sigma^2+P_J$. This equilibrium saddle point value is the capacity of the Costa channel in the presence of the jammer. Finally, from~\eqref{eq:costa:game:saddle:point} and~\eqref{eq:capacity:SI:jammer}, the surprising result that the capacity of the Costa channel in the presence of jamming is equal to that of the Side Information channel with a jammer is seen to be true.
\section{Conclusion}\label{sec:conclusion}
We examined the problem of the capacity of an AWGN channel with additive i.i.d. Gaussian state in the presence of an additive jamming adversary, where the encoder and jammer had non-causal access to the state. The user-jammer interaction was modeled as a zero sum game and the capacity of the channel, which was defined as the Nash equilibrium of this game, was established. At equilibrium, the user strategy was dirty paper coding and the jammer strategy was i.i.d. Gaussian jamming. Interestingly, however, this equilibrium jamming signal was independent of state. A surprising result, that the capacity was unchanged even if, in addition to the encoder and the jammer, the decoder too had non-causal knowledge of the state, was proved.
\section*{Acknowledgment}
The authors thank Sibiraj B. Pillai for fruitful discussions.
The work  was supported in part by the
Bharti Centre for Communication, IIT Bombay, a grant from the Information Technology Research Academy, Media Lab Asia, to IIT Bombay and TIFR, a grant from the Department
of Science and Technology, Government of India, to IIT Bombay, and the Ramanujan Fellowship from the Department of Science and Technology, Government of India, to V. M. Prabhakaran.
\addcontentsline{toc}{section}{Acknowledgment} 

\bibliographystyle{IEEEtran}

\bibliography{IEEEabrv,References}
\end{document}